\documentclass{llncs}
\usepackage[latin1]{inputenc}
\usepackage{amsmath,amssymb}
\usepackage{amsfonts}
\usepackage{upgreek}
\renewcommand{\subparagraph}{}
\usepackage{url}
\usepackage{multirow}  
\usepackage{wrapfig}   
\usepackage{graphicx}  
\usepackage{subfig}    
\usepackage{rotating}  
\usepackage{ulem}
\usepackage{enumerate}   
\usepackage{algorithm}
\usepackage{algorithmic} 
\usepackage[dvips,bookmarks=true,
            plainpages=false,naturalnames=true,
            colorlinks=true,
            linkcolor=black,citecolor=black,urlcolor=black]
        {hyperref}

\newenvironment{keywords}{
       \list{}{\advance\topsep by0.35cm\relax\small
       \leftmargin=1cm
       \labelwidth=0.35cm
       \listparindent=0.35cm
       \itemindent\listparindent
       \rightmargin\leftmargin}\item[\hskip\labelsep
       \bfseries Keywords:]}
{\endlist}

\title{Best Effort and Practice Activation Codes}
\author{Gerhard de Koning Gans\inst{1} and Eric R. Verheul\inst{1,2}}
\institute{Institute for Computing and Information Sciences \\
  Radboud University Nijmegen\\
  P.O. Box 9010, 6500 GL \ Nijmegen, The Netherlands\\
  \email{\{gkoningg,eric.verheul\}@cs.ru.nl} \and
  PricewaterhouseCoopers Advisory\\
  P.O. Box 22735, 1100 DE \ Amsterdam, The Netherlands\\
  \email{eric.verheul@nl.pwc.com}
}

\begin{document}
\normalem

\maketitle

\newcommand{\ACS}{\mathcal{S}}
\newcommand{\ACA}{\mathcal{A}}
\newcommand{\ACL}{\lambda}
\newcommand{\ACN}{\mathcal{N}}
\newcommand{\ACP}{\mathcal{P}}
\newcommand{\ACK}{\mathcal{K}}

\newcommand{\Prob}[1]{\mathrm{Pr}[#1]}
\newcommand{\Rand}[1]{\mathrm{Rand}(#1)}

\newcommand{\HMAC}[1]{\mathrm{HMAC}(#1)}
\newcommand{\BEPAC}{BEPAC }

     \newcommand{\lsup}[2]{%
        \ensuremath{{}^{#2}\!{#1}}}

\begin{abstract}
Activation Codes are used in many different digital services and known by many different names including voucher, e-coupon and discount code. In this paper we focus on a specific class of ACs that are short, human-readable, fixed-length and represent value. Even though this class of codes is extensively used there are no general guidelines for the design of Activation Code schemes. We discuss different methods that are used in practice and propose BEPAC, a new Activation Code scheme that provides both authenticity and confidentiality. The small message space of activation codes introduces some problems that are illustrated by an adaptive chosen-plaintext attack (CPA-2) on a general 3-round Feistel network of size $2^{2n}$. This attack recovers the complete permutation from at most $2^{n+2}$ plaintext-ciphertext pairs. For this reason, BEPAC is designed in such a way that authenticity and confidentiality are independent properties, i.e. loss of confidentiality does not imply loss of authenticity.
\end{abstract}


\section{Introduction}
This paper introduces Activation Codes (ACs) as a generic term for codes that are used in many different digital services. They are known by many different names including voucher, e-coupon and discount code. The common properties of these codes are that they need to be short, human-readable, have a fixed length and can be traded for economic benefit. There are schemes~\cite{blundo2002lpg,blundo2005secure,cimato-online} that include all kinds of property information in the code itself or include digital signatures~\cite{kumar-sales,jakobsson1999secure}. This makes the codes unsuitable for manual entry and thus for printing on products, labels or receipts. The focus of this paper is on ACs that can be printed and manually entered such as the AC that is printed on a receipt in Figure~\ref{pic:acexample}. In this case the customer can enter the AC `\texttt{TY5FJAHB}' on a website to receive some product. We propose a scheme called BEPAC to generate and verify this class of ACs. BEPAC is an acronym for Best Effort and Practice Activation Codes. Here `\emph{best practice}' covers the use of a keyed hash function to satisfy authenticity and `\emph{best effort}' covers the use of a Feistel network to satisfy confidentiality.

\begin{wrapfigure}{r}{0.45\textwidth}
  \vspace{-20pt}
    \includegraphics[width=0.44\textwidth]{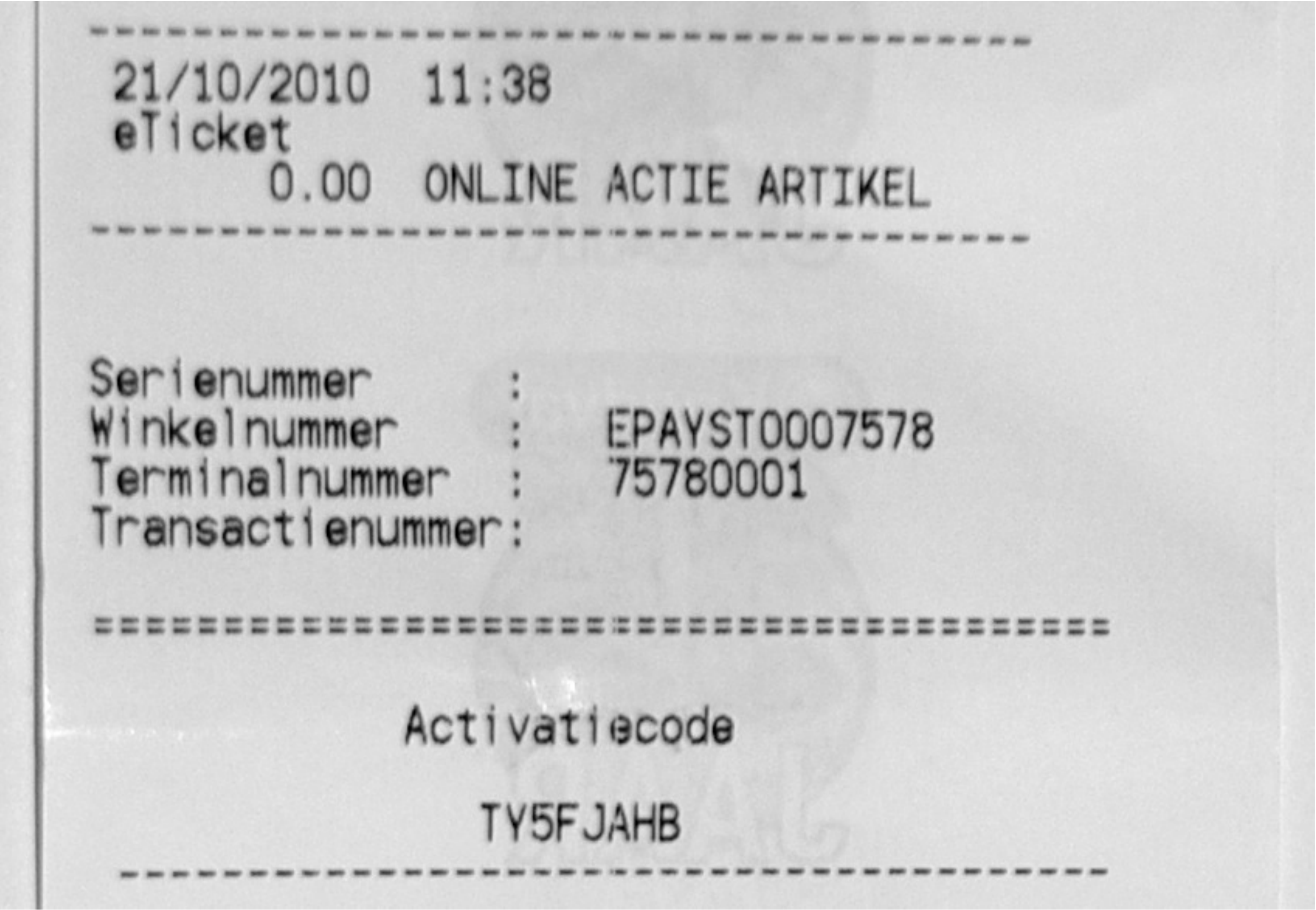}
  \caption{Activation Code}\label{pic:acexample}
  \vspace{-20pt}
\end{wrapfigure}
Security plays an important role in the design of an AC system because of the economic value it represents. A system breach could result in big financial losses.
Nevertheless, to the best of our knowledge, there are no guidelines on the design of secure AC systems that consider the previously mentioned properties.
Despite the lack of general guidelines for good practice, ACs are extensively used.
This underlines the need for a proper AC scheme that relies on elementary, well-studied, cryptographic primitives to provide authenticity and confidentiality.
First, we discuss some examples that illustrate the need for a scheme that provides both confidentiality and authenticity. Then, we give a general definition of an AC scheme and use it as a reference throughout this paper.

\paragraph{Our Contribution}
This paper addresses some known methods that are used to generate ACs and proposes BEPAC, an AC scheme that combines best effort with best practice. BEPAC is based on well-studied cryptographic primitives to guarantee unique and authentic codes that provide a satisfactory level of confidentiality. Confidentiality is obtained by a Feistel construction. The Feistel construction has weak theoretical bounds when it is used on a small domain, therefore we do not rely on it for authenticity. We use the work of Black and Rogaway~\cite{black2002caf} on small domain encryption and make some small changes to achieve confidentiality. A practical attack on a general 3-round Feistel network is presented to demonstrate the weak bounds of the Feistel construction. For authenticity, we solely rely on a keyed-hash message authentication code (HMAC) of the serial number, where the size of this HMAC determines the probability of successfully guessing a valid AC. This separated design allows a separate analysis of both confidentiality and authenticity. An advantage of this approach is that authenticity is not automatically compromised when confidentiality is broken. Finally, a BEPAC solution fits on a smart card and therefore allows AC generation and clearance to be performed in a controlled environment.

\paragraph{Related Work}
Black and Rogaway~\cite{black2002caf} propose the Generalized Feistel Cipher (GFC) as a solution to small domain encryption. This elegant solution can be used to construct a permutation on any finite domain. In the BEPAC scheme we use their method in a slightly adapted way and solely to provide confidentiality.
Black and Rogaway provide an adapted proof of Luby~\cite{luby1996pac} to prove secrecy of a 3-round Feistel network.
However, in their example configuration, the single DES round function does not give the $3 \times 56$-bit security since single DES can be broken by exhaustive key search~\cite{quisquater2005eks}. Moreover, in this setting it can be broken round by round which actually means that we have 58-bit security. Bellare et al.~\cite{spies2009} propose 128-bit AES as a pseudo-random function which drastically increases the effort needed to break one round of the Feistel network by brute-force. However, the adaptive chosen-plaintext attack (CPA-2) on a 3-round Feistel network presented in this paper shows that the key length of the pseudo-random function used in each round does not have any influence on the attack complexity. Research on Feistel networks~\cite{patarin2004security,patarin91new,naor1999cpp,luby1996pac,knudsen2008security,patarin1998feistel} has resulted in theoretical security bounds. Feistel constructions of six or more rounds are secure against adaptive chosen plaintext and chosen ciphertext attacks (CPCA-2) when the number of queries $m \ll 2^{n}$, see~\cite{patarin2004security,knudsen2008security}.
There is more related literature on the design of ACs, but to the best of our knowledge there are no proposals for the class of AC schemes that we discuss in this paper.
Blundo et al.~\cite{blundo2002lpg,blundo2005secure,cimato-online} introduce an e-coupon which is 420 bytes in size. This scheme uses a message authentication code (MAC) over some characterizing data like the identity of the manufacturer, name of the promoted product, expiry date etc. The resulting e-coupons contain valuable information but are too large to be entered manually by a user. In the work of Kumar et al.~\cite{kumar-sales} and Jakobssen et al.~\cite{jakobsson1999secure} a coupon is basically a digital signature which also means that they describe relatively large e-coupons.
Chang et al.~\cite{chang2006secure} recognize the problem of efficiency and describe a scheme that is more suitable for mobile phones that have less processing power. On the one hand, they circumvent the use of public key cryptography which reduces the computational complexity, but on the other hand, their scheme describes relatively long codes. None of the schemes described previously satisfy the requirement of short codes that can be entered manually.

Matsuyama and Fujimura~\cite{matsuyama1999distributed} describe a digital ticket management scheme that allows users to trade tickets. The authors discuss an account based and a smart card based approach and
try to treat different ticket types that are solely electronically circulated. This in contrast with BEPAC which focuses on codes that can easily be printed on product wrappings. Our intention is not to define a trading system where ACs can be transfered from one person to another. Such contextual requirements are defined in RFC 3506 as Voucher Trading Systems (VTS). Terada et al.~\cite{terada2000copy} come with a copy protection mechanism for a VTS and use public key cryptography. This makes the vouchers only suitable for electronic circulation. Furthermore, RFC 3506 and~\cite{matsuyama1999distributed} do not discuss methods on how to generate these vouchers securely. We propose the BEPAC scheme in order to fill this gap.


\section{AC Scheme Selection}\label{sec:constructions} 
This section first discusses some examples of AC systems. Then, two different approaches to set up an AC scheme are discussed and their main drawbacks are visited. After this, the Generalized Feistel Cipher of~\cite{black2002caf} is introduced in Section~\ref{secsmalldom} which has some useful concepts that we use in our scheme. The focus of an AC scheme design lays on scalability, cost-efficiency and off-line use. Finally, forgery of ACs should be hard, an adversary is only able to forge ACs with a very small predefined probability.

\paragraph{Examples of Activation Codes}
First, we discuss some examples of ACs in real life. A good first example is the \emph{scratch prepaid card} that is used in the telecommunication industry. To use a prepaid card, the customer needs to remove some foil and reveals a code that can be used to obtain mobile phone credits.
Then we have the \emph{e-coupon} which is a widely used replacement for the conventional paper coupon. An e-coupon represents value and is used to give financial discount or rebate at the checkout of a web shop.
The last example is a \emph{one-time password} that gives access to on-line content. This content should only be accessible to authenticated people who possess this unique password; think of sneak previews of new material or software distribution.

All the aforementioned examples use unique codes that should be easy to handle, that is, people should be able to manually copy ACs without much effort.
At the same time, it should be impossible for an adversary to use an AC more than once or autonomously generate a new valid AC.
Altogether, ACs are unique codes that have to guarantee authenticity. An AC system provides authenticity when an adversary is not able to forge ACs. It is a misconception that authenticity only is enough for an AC scheme. In the end, most AC systems are used in a competitive environment. When vendor $A$ starts a campaign where ACs are used to promote a product and provide it for free to customers, then it is not desirable that vendor $B$ finds out details about this campaign like the number of released ACs. Other sensitive details might be the value that different ACs represent or the expiry date of ACs. It is for this reason that we need confidentiality, which means that an adversary is not able to recognize patterns or extract any information from a released AC.

Many systems use codes that need to provide the properties discussed above. In this paper we refer to all these codes as ACs. By an AC scheme $\ACS$ we indicate a tuple $(\ACA,\ACN,\ACP,\ACL)$ where $\ACA$ is the size of the used alphabet, $\ACN$ is the number of desired ACs, $\ACP$ determines the probability $P = \frac{1}{\ACP}$ of an adversary guessing an AC and $\ACL$ is the length of the ACs. 

\subsubsection{Database Approach}
The database approach is very straightforward and consists of a database that contains all the released ACs and their current status. The generation of new AC entries is done by a pseudo-random function. When a customer redeems an AC, its status is set to `used'.
An advantage of the database approach is that the randomness of the ACs is directly related to the randomness of the pseudo-random generator. So, it is important to select a good pseudo-random generator, e.g. a FIPS certified one. 
On the other hand, the protection of this valuable data is still a problem. For instance, if an attacker manages to add entries to the database or is able to change the record status to `unused', it will be hard to detect this fraud in time. Also, it is necessary to check any new AC against all existing entries since there might be a collision. As a consequence, access to the complete set of ACs is needed on generation of new ACs.

\subsubsection{Block Cipher Approach}
Another approach is to use a block cipher that gives a random permutation $F:\{0,1\}^{n}\rightarrow\{0,1\}^{n}$ from serial numbers to ACs. The provider maintains a counter $i$ to keep track of the number of generated ACs. This way the authenticity of a serial number can be checked since only ACs that decrypt to a serial $< i$ are valid. A disadvantage of this method is the size of the resulting AC which is 128-bits for AES and 64-bits for 3DES. For AES, this results in a string of about 21 characters when we use numbers, upper- and lowercase characters in our alphabet. For 3DES, this is about 11 characters. Smaller block ciphers do exist, like \textsc{Katan}~\cite{de2009katan} which is 32-bits, but are not very well-studied. Furthermore, block ciphers force ACs to have a length that is a multiple of the block size $b$. An alternative could be the concept of elastic block ciphers~\cite{cook2007elastic} which is an extended scheme where variable message sizes are allowed as input. Moreover, this scheme uses well-studied block ciphers. Still, the minimal size of a plaintext message is the block size $b$ of the incorporated block cipher. So, this does not give any advantage and is still too large for our target, which is roughly 20 to 50-bit codes.

\subsection{Small Domain Ciphers}\label{secsmalldom} 
Black and Rogaway introduce the Generalized Feistel Cipher (GFC) in~\cite{black2002caf}. The GFC is designed to allow the construction of arbitrary domain ciphers. Here, arbitrary domain means a domain space that is not necessarily $\{0,1\}^{n}$. For ACs we want to use a small domain cipher where the domain size can be customized to a certain extent, therefore we look into the proposed method in~\cite{black2002caf}. Before we describe the Generalized Feistel Cipher, we briefly visit the basic Feistel construction.

\subsubsection{Feistel Network}\label{secfeistel}
\begin{wrapfigure}{r}{0.17\textwidth}
  \vspace{-20pt}
    \includegraphics[width=0.16\textwidth]{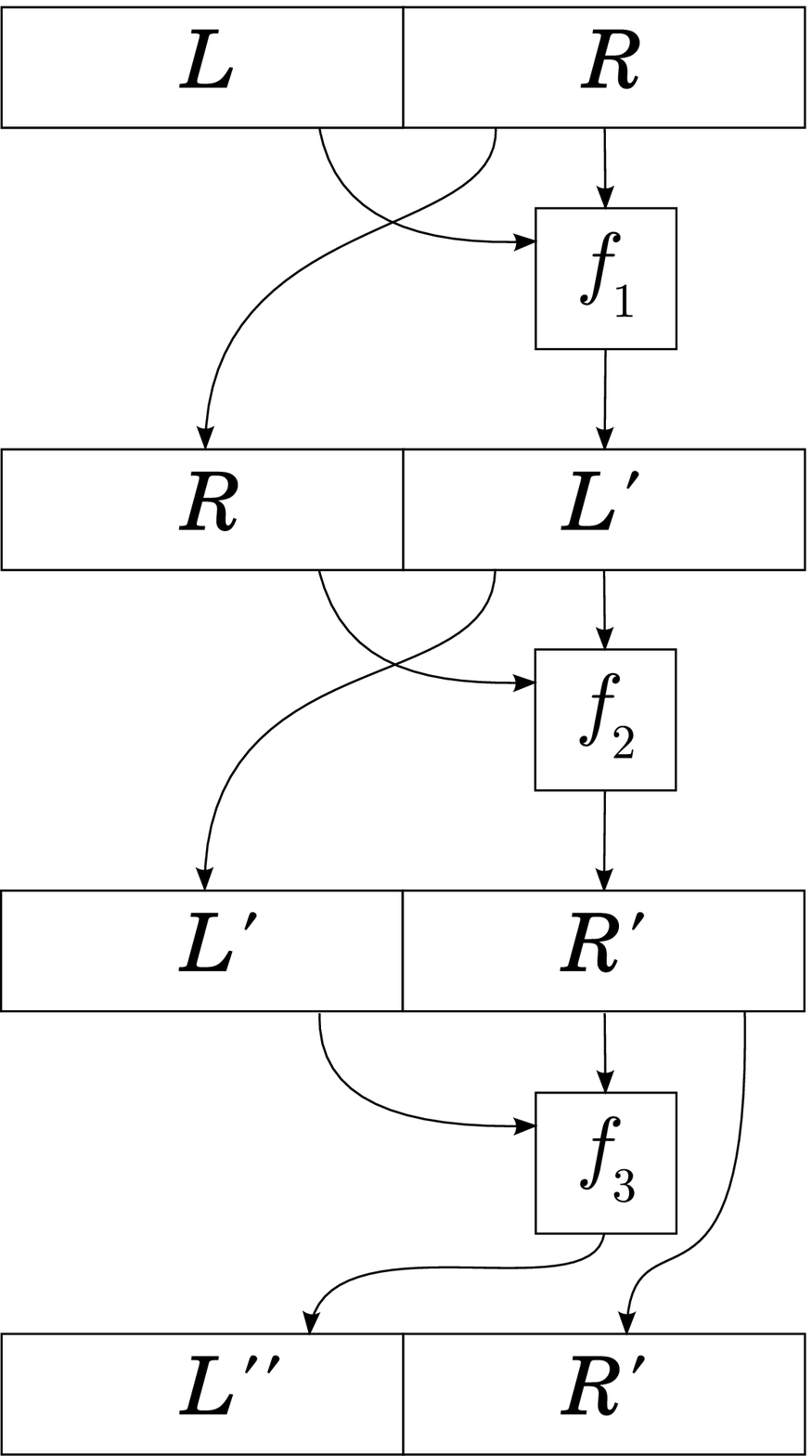}
  \caption{Feistel}\label{feistel:3rnd}
  \vspace{-40pt}
\end{wrapfigure}

A Feistel network~\cite{luby1988cpp} is a permutation that takes an input $x$ of size $2n$, then performs a number of rounds $r$ with round functions $f_{1},...,f_{r}$, and finally delivers an output $y$ of size $2n$. The input is split into two blocks $\left\langle L,R\right\rangle \in \{0,1\}^{2 \times n}$. As shown in Figure~\ref{feistel:3rnd}, every right block is input to a round function $f_i$. The output of this function is combined with the left block and becomes the new right block, e.g. $L' = f_{1}(R) + L$ for GFC. The original right block becomes the new left block. For the ease of decryption the last output blocks are swapped in case of an odd number of rounds (which is the case in Fig.~\ref{feistel:3rnd}).

\subsubsection{Generalized Feistel Cipher}
The GFC of Black and Rogaway~\cite{black2002caf} was introduced to handle flexible domain sizes. Take for example an encryption $E: 5^{14} \rightarrow 5^{14}$ which is not a domain that is captured by standard block cipher algorithms.
The BEPAC scheme borrows some of the ideas of GFC to be able to construct arbitrarily sized AC configurations.

In GFC the left and right block of the Feistel network are ``similarly sized'' which means that their domain size may deviate a little. For the particular case of ACs we have looser restrictions on the arbitrariness of our domain and we can increase the guessing probability $\ACP$ to influence the domain size. As a consequence, the system parameters of BEPAC can be chosen such that the left and right block are equally sized.

An obvious way to use the Feistel network is to create a pseudo-random permutation $F: \mathcal{K} \times \mathcal{M} \rightarrow \mathcal{M}$ where $\mathcal{K}$ and $\mathcal{M}$ are the key space and message space respectively. To generate ACs, we take as plaintext an index $i$ and use the resulting ciphertext as AC $\alpha$. In order to check $\alpha$ the provider keeps track of the last index $i$ and considers a given $\alpha$ valid when $F^{-1}(\alpha) \leq i$. This construction guarantees:
\begin{enumerate}
\item Collision-freeness, since $F$ is a permutation.
\item Valid serial numbers, they cannot be predicted since $F$ is a pseudo-random permutation.
\end{enumerate}

As Black and Rogaway already conclude in~\cite{black2002caf} the Generalized Feistel Cipher has weak security bounds when used in applications where the message space is roughly from $k = 2^{30}$ up to $k = 2^{60}$. This suggests that our second argument might not be that strong.

Also, the serial $i$ is kept secret and one might argue that this presumes unforgeability. However, the way $i$ is embedded allows an adversary to make useful assumptions about $i$ since the ACs are generated using consecutive numbers. In the GFC, the left block $L$ and right block $R$ are initiated as follows:
$$L=i\bmod 2^{n}, \qquad R=\lfloor i / 2^{n} \rfloor$$ 
Here, $L$ represents the least significant bits of $i$, and the successor of every $i$ always causes a change in $L$. On the contrary, when $i$ is sequentially incremented, the value of $R$ changes only once every $2^{n}$ times. This way, the first $2^{n}$ ACs are generated with $R=0$, the second $2^{n}$ ACs with $R=1$, etc.
The problems that this little example already points out are further explained in the next section.

\section{Feistel Permutation Recovery using CPA-2}\label{sec:genfeistelprob}
In this section we present a practical attack on a three-round Feistel construction in order to illustrate the problem of choosing a small number of rounds and using a serial embedding as suggested by Black and Rogaway~\cite{black2002caf}.


\begin{theorem}
Consider a three-round 2n-bit Feistel construction. Then there exists an algorithm that needs at most $2^{n+2}$ adaptive chosen plain-/ciphertext pairs to compute any ciphertext from any plaintext and vice versa without knowledge of the secret round keys and regardless the used key length.
\end{theorem}

\begin{proof}
The two ciphertext blocks are defined in terms of the plaintext blocks as follows:

\begin{equation}\label{eq_rp}
\begin{array}{lll}
R'  & = & f_{2}(f_{1}(R) + L) + R\\
L'' & = & f_{3}(f_{2}(f_{1}(R) + L) + R) + f_{1}(R) + L\\
    & = & f_{3}(R') + f_{1}(R) + L\\
\end{array}
\end{equation}
Note that $L''$ uses $R'$ as input to $f_{3}(\cdot)$.
With $L_{i}$ we denote $L=i$ and similarly $R_{j}$ denotes $R=j$. The notation $R'_{(i,j)}$ means the value of $R'$ when $L_{i}$ and $R_{j}$ are used as input blocks. We first observe that several triples $(f_{1},f_{2},f_{3})$ lead to the same permutation and show that it is always possible to find the triple with $f_{1}(0) = 0$. To this end, if we replace the triple $(f_{1},f_{2},f_{3})$ with the triple $(f'_{1},f'_{2},f'_{3})$ defined by Equation~(\ref{eq:alttriple}), this leads to the same permutation (Equation~(\ref{eq_rp})) with the desired property that $f'_{1}(0) = 0$.
\begin{equation}\label{eq:alttriple}
f'_{1}(x) = f_{1}(x) - f_{1}(0),\qquad f'_{2}(x) = f_{2}(x + f_{1}(0)),\qquad f'_{3}(x) = f_{3}(x) + f_{1}(0)
\end{equation}
So, without loss of generality we may assume that $f_{1}(0) = 0$.

Next, we describe a method to find a triple $(f_{1},f_{2},f_{3})$ with $f_{1}(0) = 0$. First, we determine $f_{2}$, then $f_{1}$ and finally $f_{3}$.

By Equation~(\ref{eq_rp}) we get $f_2$:
\begin{equation}\label{eq:f2_sol}
f_{2}(f_{1}(R_{0}) + L_{i}) = R'_{(i,0)} - R_{0} = R'_{(i,0)} \qquad \Rightarrow \qquad f_{2}(L_{i}) = R'_{(i,0)}
\end{equation}

Now, to find $f_1$ observe that $f_{1}(j)$ is a solution for $x$ in the equation $f_2(x) = R'_{(0,j)} - R_j$. However, this equation does not always have one unique solution since $f_{2}$ is a pseudo-random function. In case of multiple solutions we compare the output of successive (wrt. $x$) function inputs with the values for $f_2(x+i)$ that were found using Equation~(\ref{eq:f2_sol}). Then, the correct $x$ is the unique solution to:
\begin{equation}\label{eq:find_f2}
f_{2}(x+i) = R'_{(i,j)} - R_{j} \qquad \text{for } i = 0, \ldots, m
\end{equation}
Sometimes $m=0$ already gives a unique solution. At the end of this section we show that with very high probability $m=1$ defines a unique solution. We find:
\begin{equation}\label{eq:f1_sol}
	f_{1}(j) = x
\end{equation}

When $f_{1}$ and $f_{2}$ are both determined, $L$ and $R$ can be chosen such that every value for $R' \in \{0,\ldots,2^{n}-1\}$ is visited. Since $R'$ functions as direct input to $f_{3}$ it is possible to find all input-output pairs for $f_{3}$.
To visit every possible input value $z=0,\ldots,2^{n}-1$ find a pair $L_{i},R_{j}$ such that $R'_{(i,j)} = z$. First, find an index $x$ such that $f_{2}(x) = z - R_{j}$. If such an $x$ does not exist choose a different value for $R_{j}$. There is always a solution for $x$ since $R_{j}$ covers the whole domain of $f_{2}$. Second, derive $L_{i}$:
\begin{equation}
	f_{2}(f_{1}(R_{j}) + L_{i}) = f_{2}(x) \qquad \Rightarrow \qquad L_{i} = x - f_{1}(R_{j})
\end{equation}
Note that the determination of $L_{i}$ and $R_{j}$ does not need any intermediate queries since it is completely determined by $f_1$ and $f_2$.
Next, we query the system with $L_{i}$ and $R_{j}$ and use Equation~(\ref{eq:f2_sol}) and~(\ref{eq:f1_sol}) to compute $f_{3}$ as follows:
\begin{equation}\label{eq:detf3}
	f_{3}(x) = L''_{(i,j)} - f_{1}(R_{j}) - L_{i}
\end{equation}
This completes the solution for a triple $(f_1,f_2,f_3)$ that results in the same permutation as the Feistel construction under attack.

\paragraph{Number of Queries}
The determination of $f_{2}$ is given by Equation~(\ref{eq_rp}) and costs $2^{n}$ queries. The determination of $f_{1}$ is given by Equation~(\ref{eq:find_f2}). The probability $p$ that there exists an $x' \neq x$ for a preselected $x$ such that
\begin{equation}\label{eq:nrqueries}
\left(f_2(x) = f_2(x')\right)\;\wedge\;\left(\bigwedge_{i = 1,\ldots,m}{f_2(x+i) = f_2(x'+i)}\right)
\end{equation}
can be split into two parts. First, we have the probability $p_1$ that there is a collision $f_2(x) = f_2(x')$ with $x \neq x'$. Then, the second probability $p_2$ covers cases where a preselected position $f_2(x+i)$ has the same value as some other preselected position $f_2(x'+i)$. We take $k = 2^n$ and the two probabilities are then given by $p_1 = 1-(\frac{k - 1}{k})^{(k-1)}$ and $p_2 = \frac{1}{k}$. Now, $p = p_1 \cdot {p_2}^{m}$ because we need to multiply by $p_2$ for every other successive match. To conclude, the probability that there is an $x' \neq x$ for a Feistel construction of size $2 \cdot k$ and with $m$ successive queries,
i.e. the probability that there is no unique solution $x$ to Equation~(\ref{eq:find_f2}), is the probability
\begin{equation}\label{eq:finalqr}
p = \frac{1}{k^{m}} - \frac{1}{k^{m}} \cdot {\left(\frac{k-1}{k}\right)}^{k-1}
\end{equation}
So, $p < \frac{1}{k^{m}}$ and depending on the size $k$, $m = 1$ already gives $p$ close to zero. In practice one might sometimes need an additional query ($m=2$) or only one query ($m=0$), but on average $m=1$.
This means that the cost for determination of $f_1$ is $2\cdot2^{n}$ queries on average.
Then, the determination of $f_{3}$ is given by Equation~(\ref{eq:detf3}) and costs at most $2^{n}$ queries. As a result, the determination of $(f_1,f_2,f_3)$ has an upper bound of $2^{n+2}$ queries.

\end{proof}


\section{BEPAC Scheme}\label{sec:bepac}
In this section we propose our Activation Code Scheme called BEPAC. Its primary objective is to ensure authenticity and its secondary objective is to provide confidentiality. Confidentiality is satisfied up to the security bounds given by Black and Rogaway in~\cite{black2002caf}. In the \BEPAC scheme, loss of confidentiality does not affect the authenticity property.

The authenticity is achieved in an obvious way by the use of an HMAC which is a keyed hash function. We take the truncated HMAC $h$ of a sequence number $i$ and concatenate it to $i$ itself. For this concatenation we use an embedding $m$ like the one used by Black and Rogaway in~\cite{black2002caf} and Spies in~\cite{spies2009}. We rely on the strength of the underlying hash function which covers the \textit{best practice} part of our solution: ACs are not forgeable. The length of an HMAC is usually too long for the ease of use that is demanded for ACs. Therefore we introduce the probability $P = \frac{1}{\ACP}$ that puts a lower bound on the success rate of guessing correct ACs. We use this parameter to limit the length of the codes, i.e. $\ACP$ determines the size of the HMAC. A lower success probability for an adversary is achieved by concatenating a bigger part of the HMAC and thus results in a longer AC.

Our solution differs from encryption schemes for small domains~\cite{black2002caf,morris2009encipher} in the sense that we make a clear separation between the part that provides authenticity and the part that provides confidentiality. The latter is added as an additional operation on the embedding $m$. We use a balanced Feistel construction as proposed in~\cite{black2002caf} to create the necessary confusion and diffusion. This separation between authenticity and confidentiality is really different from an approach where the sequence number $i$ is directly fed into a Feistel construction and when it solely depends on this construction for its authenticity.
The attack in Section~\ref{sec:genfeistelprob} demonstrates that we cannot rely on a Feistel construction for authenticity when it is used on a small domain. These results form the basis of our design decision.

\subsection{AC Scheme Setup}
\begin{wrapfigure}{r}{0.40\textwidth}
  \vspace{-30pt}
  \begin{center}
    \includegraphics[width=0.39\textwidth]{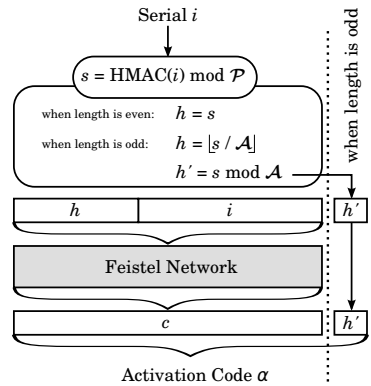}
  \end{center}
  \caption{BEPAC Scheme}\label{fig:bepac}
  \vspace{-20pt}
\end{wrapfigure}
The BEPAC scheme setup is a construction (see Fig.~\ref{fig:bepac}) where an embedding $m$ of an index $i$ and a part of $\mathrm{HMAC}(i)$ are fed into a Feistel network. Since this is a balanced Feistel network, $m$ needs to be divided into two equally sized blocks. When this is not possible a small part $h'$ of $\mathrm{HMAC}(i)$ bypasses the Feistel network and is embedded together with the cryptogram $c$ from the Feistel network to form AC $\alpha$.

The BEPAC scheme $\ACS$ is a tuple $(\ACA,\ACN,\ACP,\ACL,\omega)$ where $\ACA$ is the size of the alphabet, $\ACL+\omega$ is the length of the ACs where $\ACL$ is always even and $\omega$ is either 0 or 1. Then, $\ACN$ is the number of ACs and $\ACP$ determines the probability $P = \frac{1}{\ACP}$ of obtaining a valid AC by a random guess, e.g. $\ACP=10.000$. We assume $\ACA < \ACP$.


\begin{definition}[Valid AC Scheme]
An AC scheme $\ACS = (\ACA,\ACN,\ACP,\ACL,\omega)$ is \emph{valid} when $\ACA^{\ACL} \geqslant \ACN \times \ACP \times \ACA^{-\omega}$ holds and $\ACL$ is even.
\end{definition}

A valid AC scheme $\ACS$ can be obtained as follows:

\begin{enumerate}[(a)]
	\item\label{it:sone} The user chooses the alphabet size $\ACA$, desired number of ACs $\ACN$ and some minimal guess probability $\frac{1}{P}$.
	\item Now the minimal length $\ACL$ is calculated such that $\ACA^{\ACL} >= \ACN \times \ACP$ by taking $\ACL = \left\lceil \lsup{\log(\ACN \times \ACP)}{\ACA}\right\rceil$
	\item $|\ACA^{\ACL} - \ACN \times \ACP|$ is minimized by taking $\ACP = \lfloor \ACA^{\ACL} / \ACN \rfloor$
	\item The length $\ACL$ can be either odd or even:
		\begin{itemize}
		\item When $\ACL = 2k + 1$ and $\ACA < \ACP$ then we adjust $\ACP$ such that $\ACA$ is a divisor of $\ACP$. As a consequence, we might have a larger number of ACs $\ACN$.
		$$\ACP = \ACP - (\ACP \bmod \ACA), \qquad \ACN = \lfloor \ACA^{\ACL}/\ACP \rfloor$$
		After these operations we obtain the system $\ACS = (\ACA, \ACN, \ACP, \ACL-1, 1)$.
		\item When $\ACL = 2k$ we obtain the system $\ACS = (\ACA, \ACN, \ACP, \ACL, 0)$.
		\end{itemize}
	\item The process is repeated from step~(\ref{it:sone}) when no valid system $\ACS$ is found.
\end{enumerate}

\subsection{Generation}
This section describes the generation of new ACs once a valid AC scheme is configured. Algorithm~\ref{alg_ac_enc} contains the pseudo code for AC generation.
The plaintext is an embedding $m$ of a part of $\mathrm{HMAC}(i)$ and $i$ itself. In case of an odd AC length ($\omega=1$) a small part $h'$ of $\mathrm{HMAC}(i)$ is excluded from this embedding. The part of $\mathrm{HMAC}(i)$ that is used in $m$ is determined by $\ACP$.
$$s = \mathrm{HMAC}(i) \bmod \ACP, \qquad h = \lfloor s \times \ACA^{-\omega} \rfloor, \qquad h' = s \bmod \ACA, \qquad m = h \times \ACN + i$$
The balanced Feistel construction is defined with:
$$k = \ACA^{(\ACL / 2)}, \qquad L = m \bmod k, \qquad R = \lfloor m\;/\;k\rfloor$$
$L$ and $R$ are input blocks with size $k$ of a balanced Feistel network. We denote the output blocks after $r$ rounds by $L^\star$ and $R^\star$. When the number of rounds $r$ is even the cryptogram $c$ is given by:
$$c = R^\star \times k + L^\star$$
When $r$ is odd, the left and right block are swapped and the cryptogram $c$ is given by:
$$c = L^\star \times k + R^\star$$
This difference between odd and even is there to allow the same construction for encoding and decoding.
Finally, the activation code $\alpha$ is given by:
$$\alpha = c \times \ACA^{\omega} + {\omega}h'$$

\subsection{Verification}
This section describes the verification of previously generated ACs for a valid AC scheme. Algorithm~\ref{alg_ac_dec} contains the pseudo code for AC verification.
Given an AC $\alpha$ and an AC scheme $\ACS$ the validity can be checked as follows. First compute $c$ and $h'$ from $\alpha$:
$$c = \left\lfloor \alpha / \ACA^{\omega} \right\rfloor, \qquad h' = \alpha \bmod \ACA^{\omega} $$
The balanced Feistel construction is defined with input block size $k = \ACA^{(\ACL / 2)}$. Now, the input blocks $L$ and $R$ are obtained from $c$ as follows:
$$L = c \bmod k, \qquad R = \left\lfloor c / k \right\rfloor$$

$L$ and $R$ are input blocks with size $k$ of a balanced Feistel network. We denote the output blocks after $r$ rounds by $L^\star$ and $R^\star$. In this case we want to decrypt and therefore use the round keys in reverse order. When the number of rounds $r$ is even the plaintext $m$ is given by:
$$m = R^\star \times k + L^\star$$
When $r$ is odd, the left and right block are swapped and the plaintext $m$ is given by:
$$m = L^\star \times k + R^\star$$
Now, we are able to obtain the partial HMAC $h$ and index $i$ from $m$ by:
$$h = \left\lfloor m / \ACN \right\rfloor, \qquad i = m \bmod \ACN$$
We calculate the partial HMAC $h_t$ and $h'_t$ like in the encoding, but now we use the recovered index $i$. Finally, we say that $\alpha$ is a valid AC iff $h_t = h$ and ${\omega}h'_t = {\omega}h'$.

\vspace{-15pt}
\noindent\begin{minipage}{\textwidth}
  \centering
  \begin{minipage}[t]{.45\textwidth}
    \centering
    \begin{algorithm}[H]
 
\scriptsize
\caption{\textsc{Generate}($i$)}
\label{alg_ac_enc}
\begin{algorithmic}
\STATE $k \gets \ACA^{(\ACL / 2)}$
\STATE $s \gets \mathrm{HMAC}(i) \bmod \ACP$
\STATE $h \gets \lfloor s \times \ACA^{-\omega} \rfloor$
\STATE $h' \gets s \bmod \ACA$
\STATE $m \gets h \times \ACN + i$
\STATE $L \gets m \bmod k; \; R \gets \lfloor m / k \rfloor$
\FOR {$j \gets 1$ \textbf{to} $r$}
\STATE $tmp \gets (L + f_{j}(R)) \bmod k$
\STATE $L \gets R; \; R \gets tmp$
\ENDFOR
\IF{$r$ is odd}
\STATE $c \gets L \times k + R$
\ELSE
\STATE $c \gets R \times k + L$
\ENDIF
\STATE $\alpha \gets c \times \ACA^{\omega} + {\omega}h'$
\quad\\
\quad\\
\quad\\
\quad\\
\quad\\
\quad
\end{algorithmic}
\end{algorithm}
  \end{minipage}
\quad
  \begin{minipage}[t]{.45\textwidth}
    \centering
\begin{algorithm}[H]
\scriptsize
\caption{\textsc{Verify}($\alpha$)}
\label{alg_ac_dec}
\begin{algorithmic}
\STATE $k \gets \ACA^{(\ACL / 2)}$
\STATE $c \gets \lfloor \alpha / \ACA^{\omega}\rfloor;\; h' \gets \alpha \bmod \ACA^{\omega}$
\STATE $L \gets c \bmod k;\; R = \lfloor c / k \rfloor$
\FOR {$j \gets r$ \textbf{to} $1$}
\STATE $tmp \gets (L + f_{j}(R)) \bmod k$
\STATE $L \gets R; \; R \gets tmp$
\ENDFOR
\IF{$r$ is odd}
\STATE $m \gets L \times k + R$
\ELSE
\STATE $m \gets R \times k + L$
\ENDIF

\STATE $h \gets \lfloor m / \ACN \rfloor;\;i \gets m \bmod \ACN$
\STATE $s \gets \mathrm{HMAC}(i) \bmod \ACP$
\STATE $h_t \gets \lfloor s \times \ACA^{-\omega} \rfloor$
\STATE $h'_t \gets s \bmod \ACA$

\IF{$h_t = h$ \textbf{and} ${\omega}h'_t = {\omega}h'$}
\STATE \textsc{Valid}
\ELSE
\STATE \textsc{Invalid}
\ENDIF
\end{algorithmic}
\end{algorithm}
  \end{minipage}
  \label{alg:twoalg}
\end{minipage}

\section{Example Application: Smart Card}\label{sec:smartcard}
In this section we want to give an example of an Activation Code System (ACS). In an ACS there are a few things that need to be managed. The index $i$ of the latest generated AC and the ACs that have been used so far. Since this information is highly valuable and represents financial value it must be well protected. Think of an application where the ACs are printed on prepaid cards covered by some scratch-off material. The production of these cards is a very secured and well-defined process to ensure that activation codes are kept secret during manufacturing. These cards need to have all kinds of physical properties, e.g. the AC should not be readable when the card is partly peeled off from the back. This can be achieved by printing a random pattern on top of the scratch-off foil.

At some point there is a very critical task to be executed when the ACs need to be delivered to the manufacturer. An obvious method to do this is to encrypt the list of ACs with a secret key. Later on in the process, this list of randomly generated codes needs to be maintained by the vendor who sells the scratch cards. This induces a big security threat since leakage of this list or unauthorized modification results in financial loss. Especially when it directly relates to the core business like in the telecommunications industry.

The use of a secure application module (SAM) significantly reduces this risk. A SAM is typically a tamper-resistant device, often a smart card, which is in most cases extensively tested and certified in accordance to a standard, e.g. the Common Criteria\footnote{\url{http://www.commoncriteriaportal.org}}. The elegance of the solution presented in this paper is that it can be implemented using smart cards. The supplier determines the probability $P$, number of codes $\ACN$, size of character set $\ACA$ and the key $K$ to be used. An obvious approach is to use two smart cards since the production and clearance of activation codes are very likely to happen at two different locations. Both smart cards are initialized using the same AC scheme $\ACS$ and the same key $K$. From that moment on the first one only gives out up to $\ACN$ new activation codes. The second one is used at the clearance house to verify and keep track of traded activation codes. This can be done by a sequence of bits where the $i$-th bit determines whether the $i$-th activation code has been cleared. For 1.000.000 activation codes approximately 122\,kB of storage is needed. This fits on a SmartMX card~\cite{smartmx} which is available with 144\,kB of EEPROM. Of course, multiple cards can be used if more ACs are needed.

\section{Analysis}\label{sec:analysis}
In this section we discuss the system parameters of BEPAC and decide on some minimal bounds and algorithms. We tested a 6-round BEPAC scheme for obvious flaws using the NIST random number test~\cite{rukhin2000statistical}. This test implementation also delivered the numbers in Table~\ref{tab:bepacconf} which give a good indication of the length $l$ of the codes compared to different AC scheme configurations. In the left column the desired values are given for the number of codes $\ACN$ and the guessing probability $\ACP$. We tested these different numbers for three different alphabet sizes $\ACA$.

\setlength{\tabcolsep}{0.3ex}
\begin{table}[t]
\begin{center}
\scriptsize
\begin{tabular}{|r|r||r|r|r|r||r|r|r|r||r|r|r|r|r|}
\hline
\multicolumn{2}{|c|}{\textbf{Desired}} & \multicolumn{4}{|c|}{\textbf{$\ACA=8$}} & \multicolumn{4}{|c|}{\textbf{$\ACA=20$}} & \multicolumn{4}{|c|}{\textbf{$\ACA=31$}}\\\hline
$\ACN$ 	& $\ACP$ 	& $\ACN$ 	& $\ACP (\times 10^{3})$ 		& $l$	& Bits 		& $\ACN$ 	& $\ACP (\times 10^{3})$ 		& $l$	& Bits 		& $\ACN$ 	& $\ACP (\times 10^{3})$ 		& $l$	& Bits 		\\\hline
$10^{1}$	& $10^{3}$	& $10^{1}$	  & 26,214 & 6  & 18 & $10^{1}$	& 16    & 4  & 18 & $10^{1}$      & 92,352  & 4	 & 20 \\\hline
$10^{2}$	& $10^{3}$	& $10^{2}$	  & 20,968 & 7  & 21 & $10^{2}$	& 32    & 5  & 22 & $10^{2}$	  & 286,285 & 5	 & 25 \\\hline
$10^{3}$	& $10^{3}$	& $10^{3}$	  & 16,777 & 8  & 24 & $10^{3}$	& 64    & 6  & 26 & $10^{3}$	  & 28,613  & 5	 & 25 \\\hline
$10^{4}$	& $10^{3}$	& 10.004	  & 13,416 & 9  & 27 & $10^{4}$	& 128   & 7  & 31 & $10^{4}$	  & 88,75   & 6	 & 30 \\\hline
$10^{5}$	& $10^{3}$	& 100.003	  & 10,737 & 10 & 30 & $10^{5}$	& 12,8  & 7  & 31 & $10^{5}$	  & 275,125 & 7	 & 35 \\\hline
$10^{6}$	& $10^{3}$	& 1.000.006	  & 68,719 & 12 & 36 & $10^{6}$	& 25,6  & 8  & 35 & 1.000.567	  & 27,497  & 7	 & 35 \\\hline
$10^{7}$	& $10^{3}$	& 10.001.379	  & 54,968 & 13 & 39 & $10^{7}$	& 51,2  & 9  & 39 & 10.000.012	  & 85,289  & 8	 & 40 \\\hline
$10^{8}$	& $10^{3}$	& 100.001.057	  & 43,98  & 14 & 42 & $10^{8}$	& 102,4 & 10 & 44 & 100.010.675	  & 264,368 & 9	 & 45 \\\hline
$10^{9}$	& $10^{3}$	& 1.000.010.575 & 35,184 & 15 & 45 & $10^{9}$	& 10,24 & 10 & 44 & 1.001.045.818 & 26,412  & 9	 & 45 \\\hline
\end{tabular}
\vspace{5pt}
\caption{BEPAC Configurations}\label{tab:bepacconf}
\end{center}
\vspace{-30pt}
\end{table}


\paragraph{Number of Rounds}
We found good arguments to set the minimum number of rounds to six for the BEPAC scheme. The literature shows that Feistel constructions of six or more rounds are secure against adaptive chosen plaintext and chosen ciphertext attacks (CPCA-2) when the number of queries $m \ll 2^{n}$, see~\cite{patarin2004security,knudsen2008security}. Patarin~\cite{patarin1998feistel} shows that an adversary needs at least $2^{3n/4}$ encryptions to distinguish a six-round Feistel construction from a random permutation.
A six round Feistel network sufficiently covers the risk of leaking serial number information, but this is of course a minimum.

\paragraph{Key Derivation}
In the BEPAC scheme we need different round keys for every Feistel round and another different key for the calculation of the HMAC on the serial. We propose to derive these keys from an initial randomly chosen key~\cite{barker2007recommendation} by a key derivation function (KDF). There are several definitions available for KDFs and we propose to use KDF1 which is defined in ISO\,18033-2~\cite{ISO18033-2-2006}. Recommendations for KDFs and their construction can also be found in~\cite{chen2009recommendation}. The first key that is derived is used as the key in the HMAC calculation of $h$ and $h'$ (Section~\ref{sec:bepac}). After this the round keys for the Feistel construction are successively derived.

\paragraph{Pseudo-random Functions}
Furthermore, we need to decide on the pseudo-random functions (PRFs) that are used as round functions of the Feistel network. The pseudo-randomness of the permutation defined by a Feistel network depends on the chosen PRF in each round~\cite{luby1996pac}. It is straightforward to use a cryptographic hash function since we already need a hash function for the HMAC~\cite{pub198198} calculation and it keeps our AC scheme simple.

\paragraph{Hash Function}
In the end, the BEPAC scheme is solely based on a single cryptographic hash function. We follow the secure hash standard FIPS\;180-3~\cite{pub2008180} and propose to use an approved hash function like SHA-256.

\section{Conclusions} \label{conclusions}
In this paper we have introduced activation codes (ACs), short codes of fixed length, that represent value. These ACs should be scalable, cost efficient and forgery resistant. In the literature, several solutions~\cite{blundo2002lpg,blundo2005secure,cimato-online,kumar-sales,jakobsson1999secure,chang2006secure,matsuyama1999distributed,terada2000copy,matsuyama1999distributed} handle digital coupons or tickets that are somehow reminiscent to our notion of ACs. The difference is that most solutions use public key cryptography or other means that result in lengthy codes. In fact, these solutions come closer to some extended notion of digital cash and are not meant to give a solution on the generation of ACs. To the best of our knowledge there is no scheme that focuses on the class of ACs that are described in this paper (roughly think of 20 to 50-bit codes). Our proposed AC scheme for this class satisfies authenticity and confidentiality in a way that when confidentiality is compromised it does not automatically break authenticity and vice versa.

In order to allow a relatively small and arbitrary message space for our AC scheme we use some of the ideas of Black and Rogaway~\cite{black2002caf} in their Generalized Feistel Cipher to satisfy the confidentiality in our scheme. Several studies~\cite{patarin2004security,patarin91new,naor1999cpp,luby1996pac,knudsen2008security,patarin1998feistel} show that the security bounds of Feistel constructions are not strong enough and thus make the use of Feistel constructions in small domains questionable. To illustrate this, we have demonstrated that CPA-2 allows an adversary to recover the complete permutation from only $2^{n+2}$ plaintext-ciphertext pairs. Still, the Feistel construction is suitable for the purpose of confidentiality in our AC scheme. Since confidentiality is a secondary goal, it relaxes the demands on the security bounds. Furthermore, in BEPAC the plaintext cannot be predicted which counters the attacks from the literature. And most important of all, a Feistel construction defines a permutation on the AC domain, which means in practice that we do not have to store any additional data in order to remember which ACs are already published and which are not.

To conclude, we found satisfactory system parameters for the minimum number of Feistel rounds and we defined a method to do the key derivation for the round keys. Furthermore, we suggested a specific pseudo-random function (PRF) and hash function for concrete implementations. Finally, we have implemented the BEPAC scheme\footnote{\url{http://www.cs.ru.nl/~gkoningg/bepac}} and performed some statistical tests using the NIST Random Number Test~\cite{rukhin2000statistical}. This test did not reveal any obvious flaws.

It might be interesting for future work to create a smart card implementation of BEPAC as suggested in Section~\ref{sec:smartcard}.

\bibliographystyle{plain}

\end{document}